%% file: root.tex
\documentclass[letterpaper, 10 pt, conference]{ieeeconf}  

\IEEEoverridecommandlockouts                              

\input{preamble}

\title{\LARGE \bf
A computationally lightweight safe learning algorithm
}

\author{Dominik Baumann$^{1,2}$, Krzysztof Kowalczyk$^3$, Koen Tiels$^4$, and Pawe\l{} Wachel$^{3}$ 
\thanks{Supported by NWO VIDI 15698 and ECSEL 101007311 (IMOCO4.E).}
\thanks{$^{1}$Department of Electrical Engineering and Automation, Aalto University, Espoo, Finland
        {\tt\small dominik.baumann@aalto.fi}}%
\thanks{$^{2}$Department of Information Technology, Uppsala University, Uppsala, Sweden}%
\thanks{$^{3}$Department of Control Systems and Mechatronics, Wroc\l{}aw University of Science and Technology, Wroc\l{}aw, Poland
        {\tt\small \{krzysztof.kowalczyk, pawel.wachel\}@pwr.edu.pl}}%
\thanks{$^{4}$Department of Mechanical Engineering, Eindhoven University of Technology, Eindhoven, The Netherlands
        {\tt\small k.tiels@tue.nl}}%
}

\newcommand{\safeopt}{\textsc{SafeOpt}\xspace}
\newcommand{\colsafe}{\textsc{CoLSafe}\xspace}

\usepackage{fancyhdr}
\newcommand{\mytitle}{\textbf{Accepted final version.}
To appear in \textit{Proc. of the IEEE Conference on Decision and Control, 2023}.\\
\copyright 2023 IEEE. Personal use of this material is permitted. Permission
from IEEE must be obtained for all other uses, in any current or future
media, including reprinting/republishing this material for advertising or
promotional purposes, creating new collective works, for resale or
redistribution to servers or lists, or reuse of any copyrighted component of
this work in other works.}
\begin{document}

\maketitle
\thispagestyle{fancy}	
\pagestyle{empty}

\input{content/abstract}

\input{content/intro}
\input{content/problem}
\input{content/algo}
\input{content/eval}
\input{content/conclusions}

\addtolength{\textheight}{-12cm}   



\section*{Acknowledgements}

The authors would like to thank Bhavya Sukhija for support with the robot simulations.

\addtolength{\textheight}{+12cm}

\bibliographystyle{IEEEtran}
\bibliography{IEEEabrv,ref}

\end{document}

%% file: preamble.tex
\usepackage[english]{babel}
\usepackage{lmodern}
\usepackage[utf8]{inputenc}
\usepackage{balance}
\usepackage{xspace}
\usepackage{bbm}
\usepackage{rotating}
\usepackage{url}
\usepackage{csquotes}
\usepackage[noadjust]{cite}
\usepackage{subfig}
\usepackage{paralist}
\usepackage{multirow}
\usepackage{multicol}
\usepackage{amsmath,amssymb,mathtools}
\usepackage{booktabs}
\usepackage{xfrac}
\usepackage{array}
\usepackage{algorithm}
\usepackage[noend]{algpseudocode}
\usepackage{textcomp}
\usepackage{siunitx}
\usepackage{microtype}
\usepackage[usenames,dvipsnames]{xcolor}
\usepackage{safecolours}
\usepackage{pgfplots}
\pgfplotsset{compat=newest,unit code/.code={\si{#1}},plot coordinates/math parser=false,grid style={lightgray}, ylabel right/.style={
        after end axis/.append code={
            \node [rotate=90, anchor=north] at (rel axis cs:1,0.5) {#1};
        }   
    }}
\usepgfplotslibrary{units,external,groupplots,fillbetween}
\usetikzlibrary{positioning,angles,quotes,patterns,shapes,spy, shapes.misc,backgrounds}

\graphicspath{{./figures/}}

\usepackage{scalerel,stackengine}

\makeatletter
\newcommand{\myitem}[1]{%
\item[#1]\protected@edef\@currentlabel{#1}%
}
\makeatother

\newtheorem{theo}{Theorem}
\newtheorem{lem}{Lemma}

\newtheorem{assume}{Assumption}

\newcommand{\fakepar}[1]{\vspace{1mm}\noindent\textbf{#1.}}

\DeclareSIUnit{\belmilliwatt}{Bm}
\DeclareSIUnit{\dBm}{\deci\belmilliwatt}

\DeclareMathOperator*{\E}{\mathbb{E}}

\newcommand{\norm}[1]{\lVert#1\rVert}
\newcommand{\abs}[1]{\left\lvert#1\right\rvert}
\newcommand*\diff{\mathop{}\!\mathrm{d}}
\DeclareMathOperator*{\R}{\mathbb{R}}

\DeclareMathOperator*{\argmax}{arg\,max}

\let\originalleft\left
\let\originalright\right
\renewcommand{\left}{\mathopen{}\mathclose\bgroup\originalleft}
\renewcommand{\right}{\aftergroup\egroup\originalright}

\newcommand\figref[1]{Fig.~\ref{#1}}

\newcommand{\eg}{e.g.,\xspace}
\newcommand{\ie}{i.e.,\xspace}

\newcommand{\cf}{cf.\@\xspace}
\newcommand{\capt}[1]{\mdseries{\emph{#1}}}

%% file: content/abstract.tex
\begin{abstract}
    Safety is an essential asset when learning control policies for physical systems, as violating safety constraints during training can lead to expensive hardware damage. 
    In response to this need, the field of safe learning has emerged with algorithms that can provide probabilistic safety guarantees without knowledge of the underlying system dynamics. 
    Those algorithms often rely on Gaussian process inference. 
    Unfortunately, Gaussian process inference scales cubically with the number of data points, limiting applicability to high-dimensional and embedded systems. 
    In this paper, we propose a safe learning algorithm that provides probabilistic safety guarantees but leverages the Nadaraya-Watson estimator instead of Gaussian processes. 
    For the Nadaraya-Watson estimator, we can reach logarithmic scaling with the number of data points. 
    We provide theoretical guarantees for the estimates, embed them into a safe learning algorithm, and show numerical experiments on a simulated seven-degrees-of-freedom robot manipulator.
\end{abstract}

%% file: content/intro.tex
\section{Introduction}
\label{sec:intro}

Data-driven or learning-based control approaches have seen tremendous successes over the last years~\cite{lillicrap2015continuous,duan2016benchmarking}.
Nevertheless, many popular machine algorithms, in particular those based on deep reinforcement learning, are challenging to apply to physical systems.
One major shortcoming is the lack of theoretical guarantees, especially during exploration.
Thus, deep reinforcement learning algorithms may violate safety constraints when applied to physical systems, potentially damaging the system and endangering the surroundings.
The field of \emph{safe learning} addresses this challenge.
Based on varying degrees of assumed knowledge about the system dynamics, safe learning algorithms seek to learn (sub)optimal policies for dynamical systems while ensuring or encouraging safety~\cite{brunke2022safe}.

In this paper, we focus on a class of algorithms that ensure safety during exploration with high probability while assuming that the concrete system dynamics are unknown.
In exchange for the guarantees, the algorithms assume that a safe initial policy is known -- otherwise, the algorithm would have no possibility to start the exploration safely -- and a certain degree of regularity of the (unknown) functions that represent the constraints.
A popular algorithm of this class is \safeopt~\cite{sui2015safe,berkenkamp2021bayesian}.
Starting from the safe initial policy, \safeopt carefully explores the policy space while only exploring policies that can be classified with high probability as safe given the regularity assumptions.
In particular, \safeopt uses Gaussian processes (GPs)~\cite{williams2006gaussian} to model the constraint functions.
While GPs are a powerful tool for function estimation, they are also a bottleneck of the approach: GP inference scales cubically with the number of data points.
This scaling property limits the applicability of such algorithms for many emerging applications.
For instance, if we want mobile robots that autonomously learn control policies, learning algorithms must be implemented on embedded devices with limited computing resources.
Furthermore, even when extensive computing resources are available, cubic scaling properties limit the potential to scale to high-dimensional systems.

To address this problem, we propose \colsafe, a computationally lightweight safe algorithm that adopts the main characteristics of \safeopt but uses the Nadaraya-Watson estimator~\cite{nadaraya1964estimating,watson1964smooth} instead of GPs to approximate the constraints.
The Nadaraya-Watson estimator scales logarithmically with the number of data points~\cite{rao1996pac}.
We show that we can derive similar safety guarantees as for \safeopt and compare both algorithms in numerical simulations of a robot arm.

\fakepar{Contributions}
In this paper, we 
\begin{itemize}
    \item propose \colsafe, a novel safe learning algorithm based on \safeopt with significantly lower computational complexity;
    \item derive probabilistic safety guarantees;
    \item evaluate the algorithm in numerical experiments.
\end{itemize}

\fakepar{Related work}
For a general introduction to and overview of the safe learning literature, we refer the reader to~\cite{brunke2022safe}.
Here, we focus on methods that provide probabilistic safety guarantees without knowledge or explicit learning of the system dynamics.
Starting from \safeopt~\cite{sui2015safe,berkenkamp2021bayesian}, various algorithms have been proposed that, for instance, try to enable global exploration~\cite{baumann2021gosafe,sukhija2023gosafeopt}, make the algorithm more adaptive~\cite{konig2021applied}, or increase the effectiveness of the exploration strategy~\cite{wachi2018safe}.
All these variations have in common that they rely on GP inference.
The limited scalability has also been addressed~\cite{duivenvoorden2017constrained,sui2018stagewise}.
Here, improvements in scalability are mainly connected to better handling the discretization of the parameter space required by the standard \safeopt algorithm, and the algorithms still rely on GP inference.
Thus, all of them could be combined with the estimator and bounds we propose in this work, lowering their computational complexity.
The Nadaraya-Watson estimator has been leveraged in prior control-related works, for instance, in system identification~\cite{schuster1979contributions,juditsky1995nonlinear,ljung2006some,mzyk2020wiener}.
Nevertheless, we are not aware of any prior work leveraging the Nadaraya-Watson estimator for safe learning.

%% file: content/problem.tex
\section{Problem setting}
\label{sec:problem}

We consider a dynamical system
\[
    \diff x(t) = z(x(t),u(t)) \diff t
\]
with state space $x(t)\in\mathcal{X}\subseteq \R^n$ and input $u(t)\in\R^m$.
For this system, we want to learn a policy $\pi$, parametrized by policy parameters $a\in\mathcal{A}\subseteq\R^d$ that determines the control input $u(t)$ given the current system state $x(t)$.
The quality of the policy is evaluated by an unknown reward function $f:\, \mathcal{A}\to\R$.
In general, we can solve these kinds of problems through reinforcement learning.
However, this would also mean that during exploration, we would allow for arbitrary actions which might lead to the violation of safety constraints.
We define safety in the form of constraint functions $g_i:\,\mathcal{A}\to\R$ that should always be above zero, with $i\in\mathcal{I}_\mathrm{g}=\{1,\ldots,q\}$. 
For instance, a constraint could be the distance of a robot to an obstacle.
We can then write the constrained optimization problem as
\begin{equation}
    \label{eqn:constr_max}
    \max_{a\in\mathcal{A}} f(a) \text{ subject to }g_i(a)\ge 0\text{ for all }i\in\mathcal{I}_\mathrm{g}.
\end{equation}
Solving~\eqref{eqn:constr_max} without knowledge about the system dynamics and reward and constraint functions is generally unfeasible.
Thus, we need to make some assumptions.
First, without any prior knowledge, we have no chance to choose parameters $a$ for a first experiment for which we can be certain that they are safe.
Thus, we assume that we have an initial set of safe parameters.
\begin{assume}
A set $S_0\subset\mathcal{A}$ of safe parameters is known. That is, for all parameters $a$ in $S_0$ we have $g_i(a)\geq 0$ for all $i\in\mathcal{I}_\mathrm{g}$ and $S_0\neq\varnothing$.
\label{ass:safe_seed}
\end{assume}
In robotics, for instance, we often have simulation models available.
As these models cannot perfectly capture the real world, they are insufficient to solve~\eqref{eqn:constr_max}.
However, they may provide us with a safe initial parametrization.
Generally, the set of initial parameters can have arbitrarily bad performance regarding the reward.
Thus, if we imagine a robot manipulator that shall reach a target without colliding with an obstacle, a trivial set of safe parameters would barely move the arm from its initial position.
While this would not solve the task, it would be sufficient as a starting point for our exploration.

Ultimately, we want to estimate $f$ and $g_i$ from data.
Thus, we require some measurements from both after we do experiments.
\begin{assume}
After each experiment, we receive measurements $\hat{f}(a) = f(a) + \omega_0$ and $\hat{g}_i(a) = g_i(a)+\omega_i$ for all $i\in\mathcal{I}_\mathrm{g}$, where $\omega_i$, with $i\in\{0,\ldots,q\}$, are independent and identically distributed $\sigma$-sub-Gaussian random variables.
\label{ass:observation_model}
\end{assume}
Lastly, we require a regularity assumption about the functions $f$ and $g_i$.
\begin{assume}
    The functions $f$ and $g_i$, for all $i\in\mathcal{I}_\mathrm{g}$, are Lipschitz-continuous with known Lipschitz constant $L<\infty$.
    \label{ass:smoothness_assumption}
\end{assume}
Assuming Lipschitz-continuity is relatively common in control and the safe learning literature.
Some approaches get around assuming to know the Lipschitz constant.
However, they often require knowledge of an upper bound of the norm of $f$ and $g_i$ in a reproducing kernel Hilbert space instead.
We argue that having an upper bound on the Lipschitz constant is more intuitive.
Generally, the Lipschitz constant can also be approximated from data~\cite{wood1996estimation}.

Note that we clearly could assume individual Lipschitz constants $L_i$.
However, for simplicity, we assume a common Lipschitz constant $L$, which would be the maximum of the individual $L_i$.

%% file: content/algo.tex
\section{Safe learning algorithm}
\label{sec:algo}

\colsafe mainly follows the structure of the popular \safeopt algorithm but replaces the GPs used to model reward and constraint functions with the Nadaraya-Watson estimator.
In the following section, we show that this estimator can provide similar guarantees as~\cite{berkenkamp2021bayesian} showed for GPs.
Then, we embed the estimates in the overall algorithm.

\subsection{Nadaraya-Watson estimator}
\label{sec:estimator}

At every iteration, we receive noisy measurements of the reward and constraint functions.
From these measurements, we seek to estimate the reward function $f$ and the constraint functions $g_i$ using the Nadaraya-Watson estimator.
For notational convenience, let us introduce a selector function
\begin{equation}
    \label{eqn:selector_fcn}
    h(a,i) := \begin{cases}f(a) & \text{ if }i=0\\
    g_i(a) & \text{ if }i\in\mathcal{I}_\mathrm{g},
    \end{cases}
\end{equation}
and the set $\mathcal{I}=\{0\}\cup\mathcal{I}_\mathrm{g}$.
Then, we can get estimates of $h$ at iteration $n$,
\begin{equation}
    \label{eqn:nadaraya-watson}
    \mu_{n}(a',i) := \sum_{t=1}^n\frac{K_\lambda(a', a_t)}{\kappa_n(a')}\hat{h}_t(a,i),
\end{equation}
where
\begin{align*}
    \kappa_n(a') &\coloneqq \sum_{t=1}^nK_\lambda(a',a_t),\\
    K_\lambda(a,a') &\coloneqq \frac{1}{c_K}K\left(\frac{\norm{a-a'}}{\lambda}\right),
\end{align*}
$K(\cdot)$ is the kernel function, $\lambda$ the bandwidth parameter, $c_\mathrm{K}$ a constant, and $\hat{h}_t(a,i)$ the measurements at iteration $t$.
The kernel function needs to meet the following assumption.
\begin{assume}
\label{ass:kernel}
The kernel $K:\,\R^d\to\R$ is positive definite and bounded such that for any $v\in\R^d$ and some $c_\mathrm{K}<\infty$, $0\le K(v)\le c_\mathrm{K}$, and $K(v)=0$ for all $\norm{v}>1$.
\end{assume}
The computation time of the Nadaraya-Watson estimator scales with $\mathcal{O}(n)$ with the number of data points for a naive implementation and with $\mathcal{O}(\log(n)^d)$ if we allow for some pre-processing~\cite[Thm.~4]{rao1996pac}.
Thus, it is significantly more efficient than using GPs whose computation time scales with $\mathcal{O}(n^3)$. 

Nevertheless, we require theoretical worst-case bounds if we seek to leverage the estimate~\eqref{eqn:nadaraya-watson} for safe learning.
Errors in~\eqref{eqn:nadaraya-watson} stem from trying to estimate an unknown function from finitely many data points and the measurement noise defined in Assumption~\ref{ass:observation_model}.
We start by deriving bounds induced by the noise.
The following lemma is a slight variation of~\cite[Lem.~1]{abbasi:2011}.
\begin{lem}
    \label{lem:bound_omega}
    Let $\{v_t:\,t\in\mathbb{N}\}$ be a bounded stochastic process and $\{\omega_t:\,t\in\mathbb{N}\}$ be an i.i.d.\ sub-Gaussian stochastic process, \ie there exists a $\sigma>0$ such that, for any $\gamma\in\R$, and any $t\in\mathbb{N}$
    \begin{equation}
    \label{eqn:sub_gaussian}
        \E[\exp(\gamma\omega_t)]\le\exp\left(\frac{\gamma^2\sigma^2}{2}\right).
    \end{equation}
    For any $\eta\in\R$, define
    \[
        \omega_n(\eta)\coloneqq \exp\left(\sum_{t=1}^n\frac{\eta\omega_tv_t}{\sigma}-\frac{1}{2}\eta^2v_t^2\right).
    \]
    Then, $\E[\omega_n(\eta)]\le 1$.
\end{lem}
\begin{proof}
    Let
    \[
        D_t\coloneqq \exp\left(\frac{\eta\omega_t v_t}{\sigma} - \frac{1}{2}\eta^2v_t^2\right).
    \]
    Clearly, $w_n=D_1D_2\ldots D_n$.
    Note, that
    \begin{align*}
        \E[D_t\mid v_t]&=\E\left[\frac{\exp\left(\frac{\eta\omega_t v_t}{\sigma}\right)}{\exp\left(\frac{1}{2}v_t^2\eta^2\right)}\,\middle\vert\, v_t\right]\\
        &=\frac{\E\left[\exp\left(\frac{\eta\omega_tv_t}{\sigma}\right)\mid v_t\right]}{\exp\left(\frac{1}{2}v_t^2\eta^2\right)}.
    \end{align*}
    Hence, due to~\eqref{eqn:sub_gaussian},
    \[
        \E[D_t\mid v_t] \le \frac{\exp\left(\frac{\left(\frac{\eta v_t}{\sigma}\right)^2\sigma^2}{2}\right)}{\exp\left(\frac{1}{2}v_t^2\eta^2\right)}=1.
    \]
    Next, for every $n\in\mathbb{N}$,
    \begin{align*}
        \E[\omega_n(\eta)\mid v_n]&=\E[D_1\ldots D_{n-1}D_n\mid v_n]\\
        &= D_1\ldots D_{n-1}\E[D_n\mid v_n]\le \omega_{n-1}(\eta).
    \end{align*}
    Therefore,
    \begin{align*}
        \E[\omega_n(\eta)]&=\E[\E[\omega_n(\eta)\mid v_n]]\\
        &\le\E[\omega_{n-1}]\le\ldots\le \E[\omega_1(\eta)]\\
        &=\E[\E[D_1\mid v_1]]\le 1,
    \end{align*}
    which completes the proof.
\end{proof}
With this result, we can now derive bounds for the summation of noise terms as in~\eqref{eqn:nadaraya-watson}.
The following result is a variation of~\cite[Thm.~3]{abbasi:2011}.
\begin{lem}
    \label{lem:sum_gauss}
    Consider $v_t$ and $\omega_t$ as in Lemma~\ref{lem:bound_omega}.
    Further, let $S_n\coloneqq\sum_{t=1}^nv_t\omega_t$ and $V_n\coloneqq\sum_{t=1}^nv_t^2$.
    Then, for any $n\in\mathbb{N}$ and $0<\delta<1$, with probability at least $1-\delta$,
    \[
        \abs{S_n} \le \sqrt{2\sigma^2\log(\delta{-1}\sqrt{1+V_n})(1+V_n)}.
    \]
\end{lem}
\begin{proof}
    Without loss of generality, let $\sigma=1$.
    For any $\eta\in\R$, let
    \[
        \omega_n(\eta)\coloneqq\exp\left(\eta S_n-\frac{1}{2}\eta^2V_n\right).
    \]
    From Lemma~\ref{lem:bound_omega}, we note that for any $\eta\in\R$, $\E[\omega_n(\eta)]\le 1$.
    Let now $H$ be a $\mathcal{N}(0,1)$ random variable, independent of all other variables.
    Clearly, $\E[\omega_n(H)\mid H]\le 1$.
    Define
    \[
        \omega_n\coloneqq \E[\omega_n(H)\mid v_t,\omega_t:\, t\in\mathbb{N}].
    \]
    Then, $\E[\omega_n]\le1$ since $\E[\omega_n]=\E[\E[\omega_n(H)\mid v_t,\omega_t:\, t\in\mathbb{N}]]=\E[\omega_n(H)]=\E[\E[\omega_n(H)\mid H]]\le 1$.
    We can also express $\omega_n$ directly as
    \begin{align*}
        \omega_n&=\frac{1}{\sqrt{2\pi}}\int \exp\left(\eta S_n-\frac{1}{2}\eta^2V_n\right)\exp\left(\frac{-\eta^2}{2}\right)\diff\eta\\
        &= \frac{1}{\sqrt{2\pi}}\int\exp\left(-\frac{1}{2}(V_n+1)\eta^2+S_n\eta\right)\diff\eta,
    \end{align*}
    which further gives
    \begin{align*}
        \omega_n &= \frac{1}{\sqrt{2\pi}}\int\exp\left(-\frac{1}{2}\frac{\left(\eta-\frac{1}{1+V_n}S_n\right)^2}{(1+V_n)^{-1}}\right)\exp\left(\frac{1}{2}\frac{S_n^2}{1+V_n}\right)\\
        &= \frac{1}{\sqrt{1+V_n}}\exp\left(\frac{1}{2}\frac{S_n^2}{1+V_n}\right)\int\frac{\sqrt{1+V_n}}{\sqrt{2\pi}}\\
        &\exp\left(-\frac{1}{2}\left(\frac{\eta-\frac{1}{1+V_n}S_n}{(1+V_n)^{-\frac{1}{2}}}\right)^2\right)\diff\eta\\
        &= \frac{1}{\sqrt{1+V_n}}\exp\left(\frac{1}{2}\frac{S_n^2}{1+V_n}\right).
    \end{align*}
    Therefore, $\mathbb{P}[\delta\omega_n\ge 1]$ equals
    \begin{align*}
        &\mathbb{P}\left[\frac{\delta}{\sqrt{1+V_n}}\exp\left(\frac{1}{2}\frac{S_n^2}{1+V_n}\right)\ge 1\right]\\
        =&\mathbb{P}\left[\exp\left(\frac{1}{2}\frac{S_n^2}{1+V_n}\right)\ge \frac{\sqrt{1+V_n}}{\delta}\right]\\
        =&\mathbb{P}\left[\frac{S_n^2}{1+V_n}\ge 2\log\left(\frac{\sqrt{1+V_n}}{\delta}\right)\right]\\
        =&\mathbb{P}\left[S_n^2\ge2\log\left(\frac{\sqrt{1+V_n}}{\delta}\right)(1+V_n)\right].
    \end{align*}
    Recall now that $\E[\omega_n]\le1$.
    Hence, due to Markov's inequality, we have
    \[
        \mathbb{P}[\delta\omega_n\ge 1]\le\delta\E[\omega_n]\le\delta,
    \]
    which completes the proof.
\end{proof}
With this, we can provide the required bounds (\cf\ \cite{wachel2022decentralized}).
\begin{lem}
\label{lem:bounds}
Under Assumptions~\ref{ass:observation_model}--\ref{ass:kernel}, we have, for all $n\ge0$, $a'\in\mathcal{A}$, and all $i\in\mathcal{I}$, with probability at least $1-\delta$,
\[
    \abs{h(a',i) - \mu_n(a',i)} < \beta_n(a',i),
\]
where
\[
\beta_{n}(a',i) \coloneqq L \lambda+2\sigma \frac{\alpha _{n}(a' ,\delta )}
{\kappa _{n}(a')},
\]
and 
\begin{equation*}
\alpha _{n}( a' ,\delta ) \coloneqq
\begin{cases} 
\sqrt{\log \left( \frac{\sqrt{2}}{\delta} \right)}, & \textnormal{if } \kappa _{n}(a') \leq 1 \vspace{3pt}\\ 
\sqrt{\kappa _{n}(a') \log \left( \frac{\sqrt{1+\kappa _{n}(a')}}{\delta} \right)}, & \textnormal{if } \kappa _{n}(a') > 1.
\end{cases}
\end{equation*}
\end{lem}
\begin{proof}
    Due to Assumption~\ref{ass:kernel}, we have $\frac{1}{c_\mathrm{K}}K(\cdot)\le 1$.
    Hence, without loss of generality, we assume in the following that $K_\lambda$ is bounded by 1.
    Following Assumption~\ref{ass:observation_model}, we further have (\cf~\cite{dean2021certainty})
    \begin{align}
    \begin{split}
        \label{eqn:bounds_1}
        &\abs{\sum_{t=1}^n\frac{K_\lambda(a',a_t)\hat{h}_t(a,i)}{\kappa_n(a')}\hat{h}_t(a,i) - h(a',i)}\\
        \le &\sum_{t=1}^n\theta_t\abs{h(a_t, i)-h(a',i)} + \abs{\sum_{t=1}^n\theta_t\omega_t},
    \end{split}
    \end{align}
    where $\theta_t\coloneqq \frac{K_\lambda(a',a_t)}{\kappa_n(a')}$ and $\omega_t$ the measurement noise at iteration $t$.
    Note, that $\sum_{t=1}^n\theta_t=1$.
    Due to Assumption~\ref{ass:kernel}, if $K_\lambda(a',a_t)>0$, then $\frac{\norm{a'-a_t}}{\lambda}\le 1$.
    Therefore, if $K_\lambda(a',a_t)>0$, \cf~Assumption~\ref{ass:smoothness_assumption},
    \[
        \abs{h(a_t,i)-h(a',i)} \le L\norm{a'-a_t}\le L\lambda,
    \]
    and since the weights $\theta_t$ sum to 1,
    \[
        \sum_{t=1}^n\theta_t\abs{h(a_t,i) - h(a',i)} \le L\lambda.
    \]
    Finally, for the noise term in~\eqref{eqn:bounds_1}, observe that
    \[
        \abs{\sum_{t=1}^n\theta_t\omega_t} = \frac{1}{\kappa_n(a')}\abs{\sum_{t=1}^nK_\lambda(a',a_t)\omega_t}.
    \]
    According to Lemma~\ref{lem:sum_gauss}, this term is upper bounded, with probability $1-\delta$, by
    \begin{align*}
        &\frac{1}{\kappa_n(a')}\sigma\times\\
        &\sqrt{2\log\left(\delta^{-1}\sqrt{1+\sum_{t=1}^nK_\lambda^2(a',a_t)}\right)\left(1+\sum_{t=1}^nK_\lambda^2(a',a_t)\right)}.
    \end{align*}
    Furthermore, since $K_\lambda(a',a_t)\le 1$ (\cf~Assumption~\ref{ass:kernel}), we obtain
    \begin{align*}
        &\frac{1}{\kappa_n(a')}\abs{\sum_{t=1}^nK_\lambda(a',a_t)\omega_t}\\
        \le &\sigma\sqrt{2\log(\delta^{-1}\sqrt{1+\kappa_n(a')})}\frac{\sqrt{1+\kappa_n(a')}}{\kappa_n(a')}.
    \end{align*}
    Observe next that, if $\kappa_n(a') > 1$,
    \[
        \frac{\sqrt{1+\kappa_n(a')}}{\kappa_n(a')}<\frac{\sqrt{2\kappa_n(a')}}{\kappa_n(a')}=\frac{\sqrt{2}}{\sqrt{\kappa_n(a')}}.
    \]
    Therefore, with probability at least $1-\delta$, for $\kappa_n(a')>1$,
    \begin{align*}
        &\frac{1}{\kappa_n(a')}\abs{\sum_{t=1}^nK_\lambda(a',a_t)\omega_t}\\
        \le &\frac{2\sigma}{\kappa_n(a')}\sqrt{\kappa_n(a')\log(\delta^{-1}\sqrt{1+\kappa_n(a')}},
    \end{align*}
    whereas for $0<\kappa_n(a')\le 1$,
    \begin{align*}
        &\frac{1}{\kappa_n(a')}\abs{\sum_{t=1}^nK_\lambda(a',a_t)\omega_t}\\
        \le &\frac{\sigma}{\kappa_n(a')}\sqrt{2\log(\delta^{-1}\sqrt{1+\kappa_n(a')})}\sqrt{1+\kappa_n(a')}\\
        \le &\frac{2\sigma}{\kappa_n(a')}\sqrt{\log\left(\frac{\sqrt{2}}{\delta}\right)},
    \end{align*}
    which completes the proof.
\end{proof}
With these bounds, we can next present the safe learning algorithm.

\subsection{The algorithm}
\label{sec:alg_overview}

Starting from the initial safe seed given by Assumption~\ref{ass:safe_seed}, we can start a first experiment and receive a measurement of reward and constraint functions.
Following~\cite {berkenkamp2021bayesian}, we then leverage the functions' Lipschitz-continuity to generate a set of parameters $a$ that is safe with high probability.
For this, we first construct confidence intervals as
\begin{equation}
    \label{eqn:confidence_interv}
    Q_n(a, i) = [\mu_{n-1}(a, i)\pm\beta_{n-1}(a,i)].
\end{equation}
As the \safeopt algorithm requires that the safe set does not shrink, we further define the contained set as
\begin{equation}
    \label{eqn:contained}
    C_n(a,i)=C_{n-1}\cap Q_n(a,i),
\end{equation}
where $C_0(a,i)$ takes values in $[0,\infty)$ for all $a\in S_0$ and takes values in $\R$ for all $a\in\mathcal{A}\setminus S_0$.
This enables us to define lower and upper bounds as $l_n(a,i)\coloneqq\min C_n(a,i)$ and $u_n(a,i)\coloneqq\max C_n(a,i)$ with which we can update the safe set:
\begin{equation}
    \label{eqn:safe_set_update}
    S_n = \bigcap_{i\in\mathcal{I}_g}\bigcup_{a\in S_{n-1}}\{a'\in\mathcal{A}\mid l_n(a,i)-L\norm{a-a'}\ge 0\}.
\end{equation}
By sampling only from this set, we can guarantee that each experiment will be safe with high probability.
However, we also seek to optimize the policy. 
Thus, for a meaningful exploration strategy, we define two additional sets~\cite{berkenkamp2021bayesian}: parameters that are likely to yield a higher reward than our current optimum (potential maximizers, $M_n$), and parameters that are likely to enlarge $S_n$ (potential expanders, $G_n$).
Formally, we define them as
\begin{align}
    \label{eqn:maximizers}
    M_n&\coloneqq \{a\in S_n\mid u_n(a,0)\ge \max_{a'\in S_n}l_n(a',0)\}\\
    \label{eqn:expanders}
    G_n&\coloneqq \{a\in S_n\mid e_n(a) >0\},
\end{align}
with 
\begin{align*}
    e_n(a)\coloneqq |\{a'\in\mathcal{A}\setminus S_n\mid&\exists i\in\mathcal{I}_\mathrm{g}:\\
    &u_n(a,i)-L\norm{a-a'}\ge 0\}|.
\end{align*}
Given those sets, we select our next sample location as
\begin{equation}
    \label{eqn:next_sample}
    a_n = \argmax_{a\in M_n\cup G_n}\max_{i\in\mathcal{I}}w_n(a,i),
\end{equation}
with $w_n(a,i) = u_n(a,i) - l_n(a,i)$, which is basically a variant of the upper confidence bound algorithm~\cite{srinivas2012gaussian}.
Further, we can, at any iteration, obtain an estimate of the optimum through
\begin{equation}
    \label{eqn:optimum}
    \hat{a}_n = \argmax_{a\in S_n}l_n(a,0).
\end{equation}
The entire algorithm is summarized in Algorithm~\ref{alg:diffopt}.

Lastly, we provide the overall safety guarantees of \colsafe.
For these guarantees, we assume a finite parameter set $\mathcal{A}$.
However, heuristic extensions to continuous domains exist~\cite{duivenvoorden2017constrained}.
\begin{theo}
\label{thm:safety}
Under Assumptions~\ref{ass:safe_seed}--\ref{ass:kernel}, when following Algorithm~\ref{alg:diffopt}, we have for all $n\ge 0$, with probability at least $1-\delta$, that $g_i(a_n)\ge 0$ for all $i\in\mathcal{I}_g$, \ie the algorithm is safe.
\end{theo}
\begin{proof}
Lemma~\ref{lem:bounds} provides equivalent bounds to the ones derived in~\cite[Lem.~1]{berkenkamp2021bayesian}.
Thus, the proof follows from \cite[Lem.~11]{berkenkamp2021bayesian}.
\end{proof}

\begin{algorithm}
\small 
\caption{Pseudocode of \colsafe.}
\label{alg:diffopt}
\begin{algorithmic}[1]
\State \textbf{Input:} Domain $\mathcal{A}$, Safe seed $S_0$, Lipschitz constant $L$
\For {$n=1,2,\ldots$} 
\State Update safe set with~\eqref{eqn:safe_set_update}
\State Update set of potential maximizers with~\eqref{eqn:maximizers}
\State Update set of potential expanders with~\eqref{eqn:expanders}
\State Select $a_n$ with~\eqref{eqn:next_sample}
\State Receive measurements $\hat{f}(a_n)$, $\hat{g}_i(a_n)$, for all $i\in\mathcal{I}_\mathrm{g}$
\State Update Nadaraya-Watson estimator~\eqref{eqn:nadaraya-watson} with new data
\EndFor
\Return Best guess~\eqref{eqn:optimum}
\end{algorithmic}
\end{algorithm}

%% file: content/eval.tex
\section{Numerical experiments}
\label{sec:eval}

In the evaluation, we seek to demonstrate that \colsafe is generally capable of learning control policies for dynamical systems and that its computations are more lightweight than those of \safeopt.
Thus, we take a simulation model that has been used in earlier work on \safeopt~\cite{sukhija2023gosafeopt}, run both \colsafe and \safeopt, and compare the results.
In particular, both algorithms are supposed to learn a control policy for a simulation model of a seven-degrees-of-freedom Franka robot arm.
The goal is to let the arm of the robot reach a desired set point without colliding with an obstacle.
We consider a feedback linearization approach based on an approximate system model and design a linear quadratic regulator (LQR) for the then approximately linear system.
Since the system model is not perfect, the LQR does not achieve optimal performance.
However, it can provide us with a safe seed.
Starting from there, we seek to optimize the cost matrices of the LQR to compensate for the model mismatch.
The parameter space for this experiment is $d=2$.
For further details on the setup, we refer the reader to~\cite{sukhija2023gosafeopt}.

For \safeopt, we adopt the parameter settings from~\cite{sukhija2023gosafeopt}, \ie we use a Mat\'{e}rn kernel with $\nu=1.5$.
For \colsafe, we use the same kernel with a length scale of 0.1, set the bandwidth parameter $\lambda=0.5$, and the Lipschitz constant $L=1.75$.
To meet Assumption~\ref{ass:kernel}, we only use the output of the kernel for the Nadaraya-Watson estimator if $\norm{a-a'}<1$; else, we set it to 0.

We run \safeopt and \colsafe and compare the results after 410 episodes in \figref{fig:eval_safe_set}.
Both algorithms expand the safe set beyond the initial region and find better controller parameters.
\safeopt can explore a larger part of the state space within the 410 episodes, \ie \colsafe is more conservative.
However, \safeopt requires more time for each individual optimization step.
We show this in \figref{fig:time_comparison}.
We measure the time for updating the safe set and suggesting the next policy parameters for both algorithms.
Especially during later iterations, when the data set is larger, \safeopt requires significantly more time for each update step.
In particular, in the last iteration, \safeopt needs more than \SI{1}{\hour} to suggest the next sample point, while \colsafe needs only \SI{12.5}{\second}. 
The times were measured on a standard laptop.

\begin{figure}
    \centering
    \subfloat[\colsafe after 410 episodes.\label{sfig:our_400}]{\includegraphics[width=0.24\textwidth]{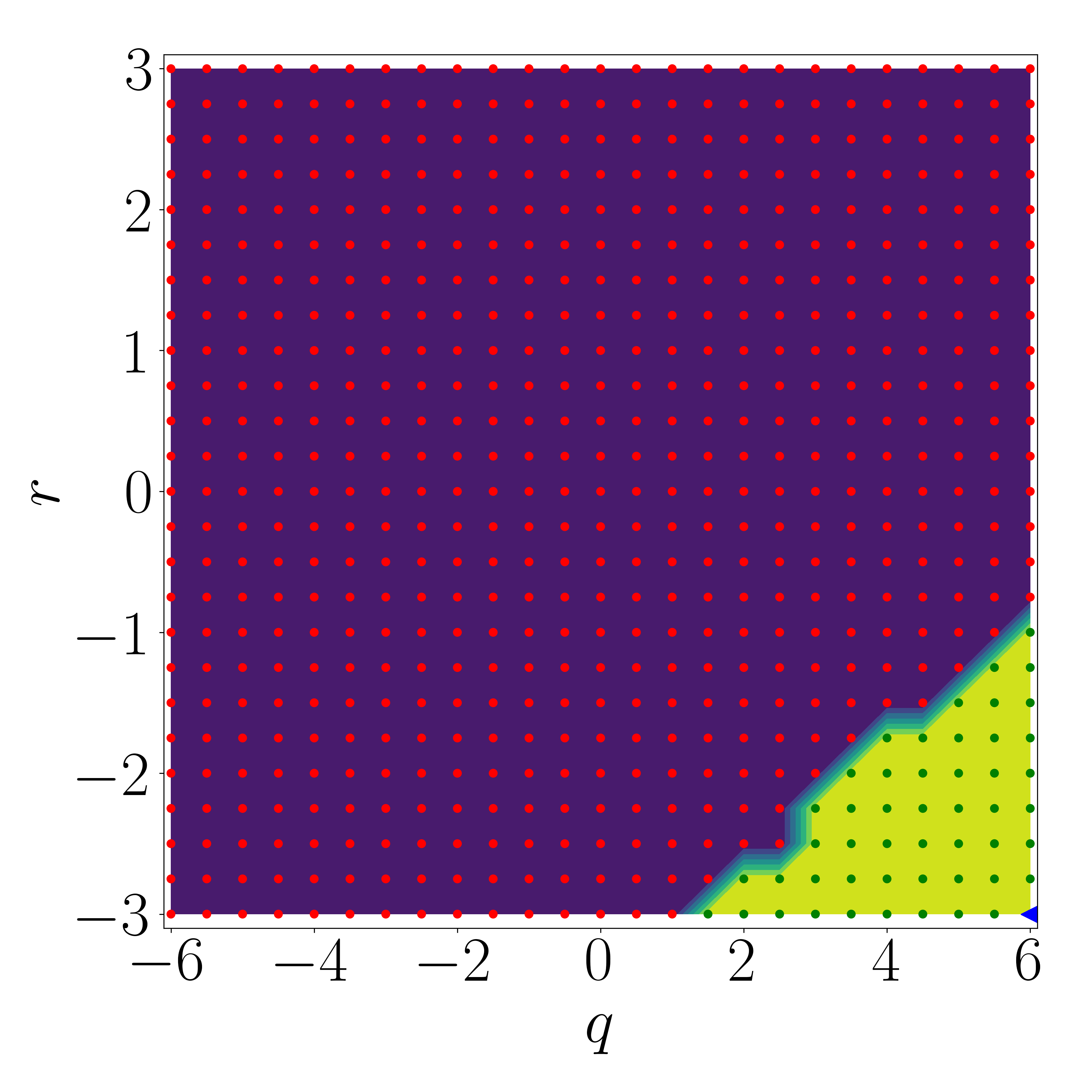}}
    \subfloat[\safeopt after 410 episodes.\label{sfig:safeopt_400}]{\includegraphics[width=0.24\textwidth]{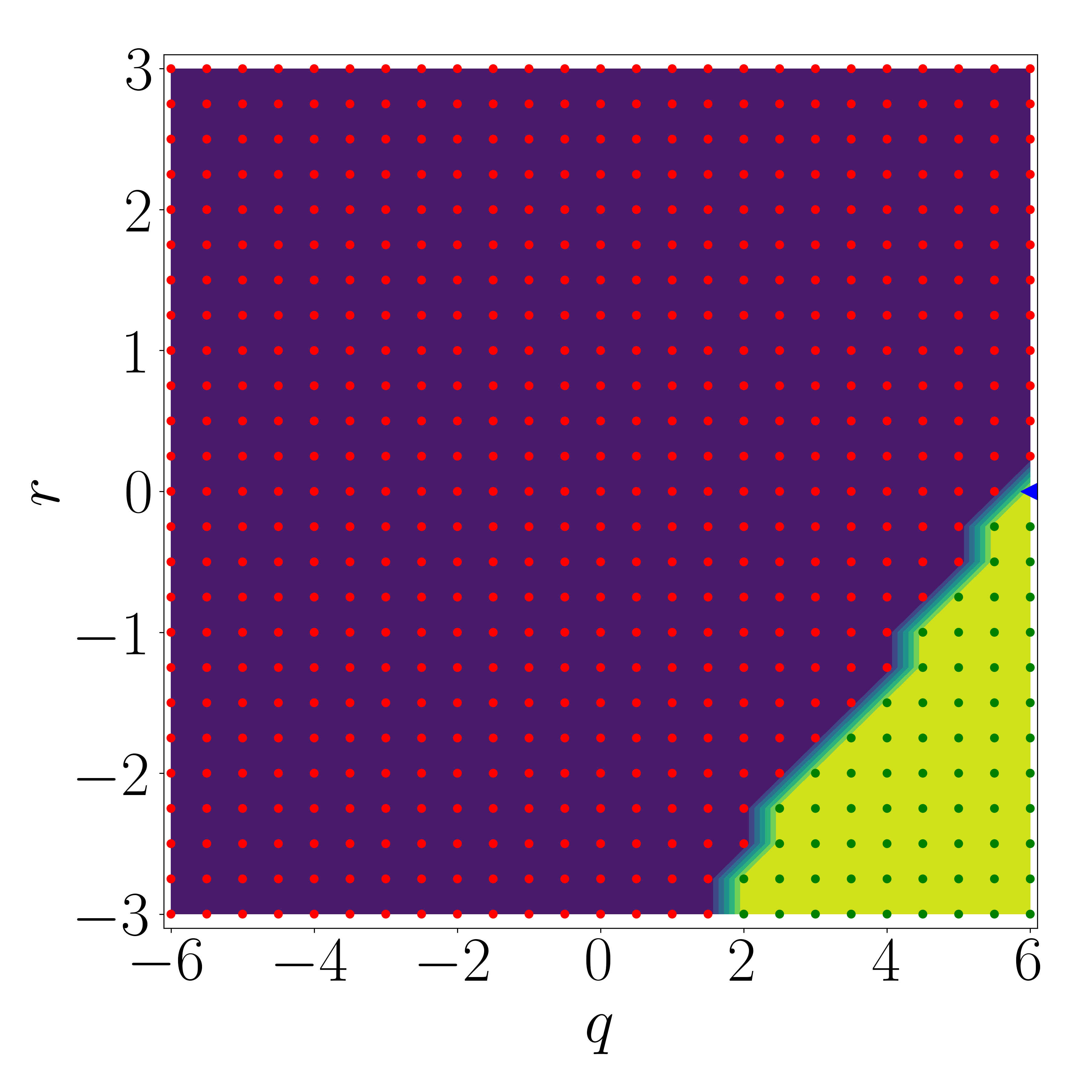}}
    \caption{Performance of \colsafe and \safeopt on the simulated Franka robot. \capt{Both algorithms can similarly explore the safe region while the computational footprint of \colsafe is significantly lower. The yellow regions mark the safe set, the blue triangle the current optimum, and $q$ and $r$ are the tuning parameters of the controller.}}
    \label{fig:eval_safe_set}
\end{figure}

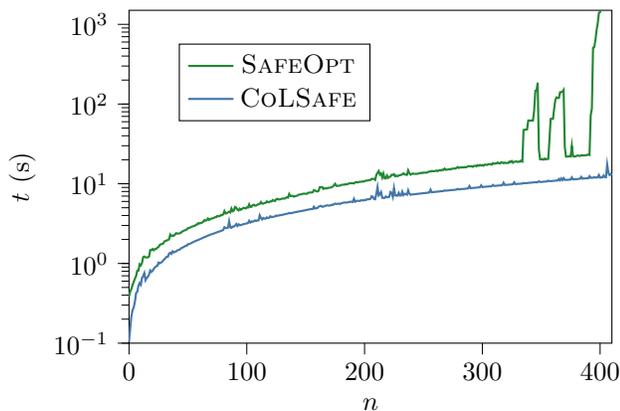
\begin{figure}
    \centering
    \input{img/time_comparison}
    \caption{Time complexity of \safeopt and \colsafe. \capt{We measure the time $t$ for updating the safe set and suggesting the next candidate point in each iteration $n$. For higher iterations, \ie more data, \safeopt requires significantly more time. Note the logarithmic scaling of the $y$-axis.}}
    \label{fig:time_comparison}
\end{figure}

This evaluation suggests an interesting trade-off.
While the computations for \colsafe are computationally cheaper, the bounds are more conservative, and thus, only a smaller safe set can be explored.
That is, when computing resources are not a concern, and the dimensionality and expected number of data samples are relatively low, but we need to explore the largest possible safe set to increase the chance of finding the global optimum, \safeopt is the better choice.
However, for higher dimensional systems and restricted computing resources, \eg when learning policies on embedded devices, \colsafe is superior.

%% file: img/time_comparison.tex
\begin{tikzpicture}

\definecolor{darkgray176}{RGB}{176,176,176}
\definecolor{darkorange25512714}{RGB}{255,127,14}
\definecolor{lightgray}{RGB}{211,211,211}
\definecolor{steelblue31119180}{RGB}{31,119,180}

\begin{axis}[
tick align=outside,
tick pos=left,
x grid style={darkgray176},
xmin=0, xmax=410,
xtick style={color=black},
y grid style={darkgray176},
height=0.25\textheight,
width = 0.45\textwidth,
ymin=1e-1, ymax=1500,
ytick style={color=black},
xlabel=$n$,
ylabel=$t$,
ymode=log,
y unit=\si{\second},
unit markings=parenthesis,
legend style={at={(0.5,0.9)}}
]
\addplot [thick, tol-colour3]
table {%
0 0.381981134414673
1 0.442886352539062
2 0.479292154312134
3 0.533766269683838
4 0.587437629699707
5 0.644595384597778
6 0.696609258651733
7 0.812113046646118
8 0.814228773117065
9 0.977716445922852
10 0.948533058166504
11 1.01144552230835
12 1.20103430747986
13 1.2222728729248
14 1.19315028190613
15 1.1878228187561
16 1.19254612922668
17 1.21882057189941
18 1.47260713577271
19 1.47520637512207
20 1.51921772956848
21 1.46717953681946
22 1.51722264289856
23 1.518394947052
24 1.57488632202148
25 1.60861325263977
26 1.69964408874512
27 1.74337410926819
28 1.74824738502502
29 1.78172373771667
30 1.88193225860596
31 1.92468237876892
32 1.94148015975952
33 1.9716637134552
34 2.00675892829895
35 2.30126571655273
36 2.20346474647522
37 2.24323701858521
38 2.22314667701721
39 2.29586982727051
40 2.39892578125
41 2.37454915046692
42 2.3724684715271
43 2.38827872276306
44 2.42821168899536
45 2.51073527336121
46 2.54416823387146
47 2.6031641960144
48 2.62549662590027
49 2.66288352012634
50 2.77358269691467
51 2.78392744064331
52 2.83325600624084
53 2.87630033493042
54 2.91896796226501
55 2.94721031188965
56 2.9755334854126
57 3.07248520851135
58 3.08619427680969
59 3.21686100959778
60 3.1734471321106
61 3.21831035614014
62 3.25446462631226
63 3.3451235294342
64 3.50250434875488
65 3.43673825263977
66 3.48285341262817
67 3.66883182525635
68 3.52524709701538
69 3.58452606201172
70 3.67225861549377
71 3.67723965644836
72 3.72463822364807
73 3.80247163772583
74 3.80648732185364
75 3.94564414024353
76 3.90316677093506
77 3.92711448669434
78 4.05077385902405
79 4.02887487411499
80 4.14185404777527
81 4.54501986503601
82 4.23502898216248
83 4.27080678939819
84 4.37059545516968
85 4.45861864089966
86 4.37423777580261
87 4.88070440292358
88 4.51794195175171
89 4.59709143638611
90 5.11850118637085
91 5.02739191055298
92 4.7379093170166
93 4.62349224090576
94 4.87107276916504
95 4.71368002891541
96 4.82415628433228
97 4.98586583137512
98 4.96184229850769
99 4.90501427650452
100 5.05360722541809
101 5.02594757080078
102 5.50958490371704
103 5.10495853424072
104 5.14026594161987
105 5.40428996086121
106 5.44660043716431
107 5.30050277709961
108 5.34132385253906
109 5.43788766860962
110 5.54849243164062
111 5.58308672904968
112 5.90697336196899
113 5.6480176448822
114 5.61962103843689
115 5.72499060630798
116 5.76783394813538
117 6.01276898384094
118 5.8626880645752
119 5.87774348258972
120 6.02060699462891
121 5.9992413520813
122 6.2534065246582
123 6.12376284599304
124 6.06700444221497
125 6.39465498924255
126 6.48713898658752
127 6.34579658508301
128 6.35185432434082
129 6.43178296089172
130 6.50811529159546
131 6.52177762985229
132 6.50141787528992
133 6.62675333023071
134 6.77191281318665
135 6.70050597190857
136 7.4500150680542
137 7.0290961265564
138 7.01116251945496
139 7.03749966621399
140 6.88703560829163
141 6.96877813339233
142 7.03400278091431
143 7.11649894714355
144 7.28386068344116
145 7.32074189186096
146 7.36357975006104
147 7.44214272499084
148 7.50678038597107
149 7.52006149291992
150 7.48431253433228
151 7.47820281982422
152 7.68916749954224
153 7.73822784423828
154 7.78209972381592
155 7.77375102043152
156 7.8434796333313
157 8.34373116493225
158 8.23429417610168
159 8.03244876861572
160 7.99994540214539
161 8.16052865982056
162 9.06728053092957
163 9.03941035270691
164 9.10427856445312
165 8.60368585586548
166 8.81332111358643
167 8.73915386199951
168 8.75203466415405
169 8.90154242515564
170 8.89178323745728
171 8.95034599304199
172 9.08130192756653
173 9.2042920589447
174 9.59463787078857
175 9.98647665977478
176 9.43370246887207
177 9.36428236961365
178 9.46463418006897
179 9.48472118377686
180 9.59564471244812
181 9.72426390647888
182 9.73419618606567
183 9.73250126838684
184 9.88791131973267
185 9.81237483024597
186 9.98917651176453
187 10.1118760108948
188 10.1215386390686
189 10.1521260738373
190 10.1644558906555
191 10.0822951793671
192 10.4080667495728
193 10.5038774013519
194 10.3728063106537
195 10.5160055160522
196 10.5998802185059
197 10.5611135959625
198 10.7076127529144
199 10.6344447135925
200 10.7555799484253
201 10.9523606300354
202 10.9402735233307
203 11.0179150104523
204 11.1342785358429
205 11.1437385082245
206 11.5415096282959
207 11.514933347702
208 11.4971752166748
209 11.1768672466278
210 12.8768446445465
211 13.7296974658966
212 14.4513070583344
213 13.0497756004333
214 13.4596610069275
215 12.0005528926849
216 13.4013783931732
217 11.6989395618439
218 11.64133477211
219 12.6620247364044
220 11.8614916801453
221 12.1723976135254
222 12.2933232784271
223 12.9710071086884
224 12.9085395336151
225 13.1998417377472
226 12.4093236923218
227 12.6397187709808
228 12.4435873031616
229 12.6657643318176
230 12.8549721240997
231 12.8787131309509
232 13.1467480659485
233 13.5116539001465
234 13.1465699672699
235 12.9646573066711
236 12.924323797226
237 14.7426850795746
238 13.6416215896606
239 13.0005993843079
240 13.2937610149384
241 13.4429669380188
242 13.3200120925903
243 13.3193826675415
244 13.2742278575897
245 13.3603129386902
246 13.4684309959412
247 13.5543496608734
248 13.7356264591217
249 13.5983493328094
250 13.8249707221985
251 13.7738816738129
252 13.7733166217804
253 13.8009705543518
254 14.0920486450195
255 14.217000246048
256 14.2548403739929
257 14.4598460197449
258 14.3575999736786
259 14.3933401107788
260 14.4604814052582
261 14.6363346576691
262 14.3935883045197
263 14.6767082214355
264 14.7926082611084
265 14.6914813518524
266 15.0105648040771
267 14.8693783283234
268 15.4862430095673
269 14.9742026329041
270 15.0192325115204
271 15.1192893981934
272 15.2651765346527
273 15.1007814407349
274 15.4553515911102
275 15.4613196849823
276 15.5559430122375
277 15.5049643516541
278 15.7438676357269
279 16.2146942615509
280 15.8006684780121
281 15.9921598434448
282 15.9858469963074
283 16.0914676189423
284 15.9999344348907
285 16.2336182594299
286 16.0443222522736
287 16.3246257305145
288 16.5380165576935
289 16.5694787502289
290 16.5526022911072
291 16.5313348770142
292 16.2990131378174
293 16.7677638530731
294 16.8405375480652
295 16.5414247512817
296 17.122213602066
297 16.8016355037689
298 17.1973180770874
299 16.9074139595032
300 17.1720323562622
301 16.8471610546112
302 17.3994345664978
303 17.3664917945862
304 17.2116339206696
305 17.5354359149933
306 17.635632276535
307 17.6845188140869
308 17.3505661487579
309 18.1342804431915
310 17.8055520057678
311 17.9735732078552
312 18.0037858486176
313 18.5181927680969
314 18.1067128181458
315 17.9181056022644
316 18.5910401344299
317 18.4472122192383
318 18.1200976371765
319 18.4332876205444
320 18.1596059799194
321 18.2721247673035
322 18.2249238491058
323 18.9639353752136
324 18.6921918392181
325 19.2854497432709
326 18.3709478378296
327 18.870274066925
328 18.9907233715057
329 19.0506482124329
330 18.7717671394348
331 18.9811942577362
332 18.8617556095123
333 19.5512413978577
334 19.1054925918579
335 47.1974737644196
336 47.7575170993805
337 47.6094651222229
338 48.1104447841644
339 62.1436750888824
340 62.2435388565063
341 61.6639819145203
342 62.2024118900299
343 62.1187093257904
344 82.9523799419403
345 146.431714773178
346 147.930626869202
347 185.370336532593
348 27.1833686828613
349 20.3032989501953
350 20.1327452659607
351 20.4825141429901
352 20.0437288284302
353 20.7138555049896
354 20.2973530292511
355 20.3980801105499
356 20.5106229782104
357 35.9505126476288
358 64.9768302440643
359 65.4460525512695
360 65.2543981075287
361 88.0750875473022
362 113.042107343674
363 121.1988260746
364 119.139318704605
365 119.141671180725
366 140.905807256699
367 141.097989082336
368 142.771202325821
369 150.466177940369
370 29.0644063949585
371 21.8982663154602
372 21.9373278617859
373 21.9278681278229
374 22.1674864292145
375 22.0088202953339
376 29.7929587364197
377 22.2340867519379
378 21.8874745368958
379 22.07941365242
380 22.3434412479401
381 22.6450593471527
382 22.4046416282654
383 22.3194499015808
384 22.8334743976593
385 22.6867306232452
386 22.6676819324493
387 22.7599282264709
388 22.8310022354126
389 22.8965713977814
390 23.3467636108398
391 23.0086133480072
392 71.926301240921
393 87.6822602748871
394 507.939176559448
395 518.670340061188
396 628.632339000702
397 845.037073612213
398 1033.45574045181
399 1394.71212720871
400 1428.48173236847
401 1566.56114053726
402 1980.41334080696
403 2049.13705658913
404 2386.20563483238
405 2527.84512543678
406 3235.84648156166
407 3238.71007561684
408 3324.05625081062
409 3424.97129845619
410 3680.20677661896
};
\addlegendentry{\safeopt}

\addplot [thick, tol-colour1]
table {%
0 0.10520076751709
1 0.151855230331421
2 0.220064878463745
3 0.258285522460938
4 0.281539916992188
5 0.332069873809814
6 0.427529811859131
7 0.441009521484375
8 0.52625036239624
9 0.568719148635864
10 0.539925098419189
11 0.658278703689575
12 0.694517850875854
13 0.757028818130493
14 0.618085384368896
15 0.65849232673645
16 0.678224325180054
17 0.727265596389771
18 0.81632924079895
19 0.782901525497437
20 0.809870004653931
21 0.853343725204468
22 0.893494129180908
23 0.920916318893433
24 0.930198907852173
25 1.02444815635681
26 1.00405430793762
27 1.02302622795105
28 1.04353928565979
29 1.09283661842346
30 1.12532138824463
31 1.16913199424744
32 1.21705937385559
33 1.21341586112976
34 1.33291506767273
35 1.36231207847595
36 1.32075119018555
37 1.41938018798828
38 1.37679171562195
39 1.40433764457703
40 1.41229152679443
41 1.46047067642212
42 1.47609663009644
43 1.49007368087769
44 1.54611802101135
45 1.57174158096313
46 1.59591722488403
47 1.62781572341919
48 1.66417956352234
49 1.66627717018127
50 1.74621248245239
51 1.74981260299683
52 1.77361917495728
53 1.80849409103394
54 1.83905076980591
55 1.87126994132996
56 1.90127825737
57 1.92584490776062
58 1.95049095153809
59 1.98198294639587
60 2.01026368141174
61 2.04265117645264
62 2.07440328598022
63 2.10252094268799
64 2.14770603179932
65 2.16272687911987
66 2.19825887680054
67 2.23766589164734
68 2.25828456878662
69 2.3004891872406
70 2.31995296478271
71 2.34886765480042
72 2.37578821182251
73 2.40728569030762
74 2.43797826766968
75 2.47165966033936
76 2.50067543983459
77 2.52850842475891
78 2.5549099445343
79 2.58886480331421
80 2.63444995880127
81 2.84960746765137
82 2.79235649108887
83 2.77145099639893
84 2.81939840316772
85 3.38181209564209
86 2.82586169242859
87 2.79979515075684
88 2.90792989730835
89 2.89896059036255
90 2.97975015640259
91 3.12470245361328
92 2.95442152023315
93 3.10789966583252
94 3.01188659667969
95 3.02728390693665
96 3.0693781375885
97 3.10908317565918
98 3.12592339515686
99 3.16204786300659
100 3.18100309371948
101 3.21783423423767
102 3.4348316192627
103 3.25901436805725
104 3.36833024024963
105 3.33174276351929
106 3.39549326896667
107 3.47362732887268
108 3.41631507873535
109 3.45675897598267
110 3.46971774101257
111 4.01982235908508
112 3.72713565826416
113 3.5625901222229
114 3.71967887878418
115 3.64521193504333
116 3.64941120147705
117 3.78051018714905
118 3.71158766746521
119 3.74682474136353
120 3.91610407829285
121 3.81890940666199
122 3.8364040851593
123 3.86877274513245
124 3.88933730125427
125 3.99906826019287
126 3.95941352844238
127 3.9977490901947
128 4.02655363082886
129 4.04340839385986
130 4.07436633110046
131 4.1031436920166
132 4.13711261749268
133 4.17859244346619
134 4.1952166557312
135 4.22669887542725
136 4.30839610099792
137 4.26371550559998
138 4.28360891342163
139 4.32419228553772
140 4.36944961547852
141 4.39533567428589
142 4.42666292190552
143 4.45381569862366
144 4.48631978034973
145 4.51715230941772
146 4.55028867721558
147 4.56681132316589
148 4.60608148574829
149 4.62606811523438
150 4.66311645507812
151 4.70674824714661
152 4.72434759140015
153 4.7715528011322
154 4.79508256912231
155 4.81444597244263
156 4.82067680358887
157 5.15278053283691
158 4.88839411735535
159 4.94121670722961
160 4.96040296554565
161 5.17417812347412
162 5.22642493247986
163 5.28866648674011
164 5.14269161224365
165 5.15665483474731
166 5.27359127998352
167 5.28437542915344
168 5.35460090637207
169 5.35532116889954
170 5.30806684494019
171 5.33968210220337
172 5.4539909362793
173 5.41359257698059
174 5.55552864074707
175 5.63208627700806
176 5.56386113166809
177 5.60548496246338
178 5.63308358192444
179 5.56035017967224
180 5.66866517066956
181 5.72499513626099
182 5.78732204437256
183 5.82861351966858
184 5.80057787895203
185 5.80145907402039
186 5.8762104511261
187 5.89958500862122
188 5.95770740509033
189 5.9243860244751
190 5.90882229804993
191 6.44940614700317
192 5.94881916046143
193 6.07430267333984
194 6.12772536277771
195 6.0975399017334
196 6.07266402244568
197 6.20586705207825
198 6.19640350341797
199 6.16599869728088
200 6.28692650794983
201 6.23437595367432
202 6.3654351234436
203 6.41078758239746
204 6.38098239898682
205 6.46444082260132
206 6.97657179832458
207 6.87197208404541
208 6.439777135849
209 6.38933849334717
210 7.49325251579285
211 9.0806782245636
212 7.1322295665741
213 7.35053014755249
214 7.8575267791748
215 6.66469073295593
216 6.62682366371155
217 6.63166832923889
218 6.92892813682556
219 6.79439091682434
220 6.74860811233521
221 6.89300298690796
222 7.89472007751465
223 7.72919631004333
224 6.88398337364197
225 8.78753137588501
226 7.15720963478088
227 7.09432315826416
228 7.06937432289124
229 7.67181897163391
230 7.12295484542847
231 7.13601660728455
232 7.95453906059265
233 7.19767189025879
234 7.48514795303345
235 7.42982983589172
236 7.27745771408081
237 8.27923464775085
238 7.53057336807251
239 7.27624273300171
240 7.32931399345398
241 7.3470184803009
242 7.39612507820129
243 7.40321731567383
244 7.4405791759491
245 7.46711158752441
246 7.50450205802917
247 7.52532768249512
248 7.55085754394531
249 7.58660888671875
250 7.62621808052063
251 7.63328409194946
252 7.66969680786133
253 7.69872379302979
254 7.74200797080994
255 7.76614117622375
256 8.42563438415527
257 7.82314085960388
258 7.84373569488525
259 7.88279032707214
260 7.91997289657593
261 7.93718409538269
262 7.97057604789734
263 7.9849157333374
264 8.02640604972839
265 8.07337403297424
266 8.06370830535889
267 8.1268105506897
268 8.26021003723145
269 8.17233347892761
270 8.20472693443298
271 8.25612616539001
272 8.26798152923584
273 8.29387354850769
274 8.36133098602295
275 8.3535737991333
276 8.37175559997559
277 8.41226863861084
278 8.43678307533264
279 8.62230944633484
280 8.51085591316223
281 8.5211136341095
282 8.57185888290405
283 8.60109281539917
284 8.59923529624939
285 8.62984275817871
286 8.67579650878906
287 8.70328450202942
288 8.73780059814453
289 9.35332489013672
290 8.77225589752197
291 8.83908915519714
292 8.86307501792908
293 8.8797345161438
294 8.91055202484131
295 8.94035696983337
296 8.98406028747559
297 9.01199197769165
298 9.04015779495239
299 9.55227375030518
300 9.09079074859619
301 9.1240758895874
302 9.14524054527283
303 9.19232749938965
304 9.22297811508179
305 9.23929977416992
306 9.27360606193542
307 9.29348874092102
308 9.31464576721191
309 9.67607688903809
310 9.42214632034302
311 9.40108633041382
312 9.4386830329895
313 9.48010540008545
314 9.51670408248901
315 9.51959300041199
316 9.56369256973267
317 9.59090638160706
318 10.2870278358459
319 9.67306804656982
320 9.68134021759033
321 9.69163608551025
322 9.71789360046387
323 9.76897192001343
324 9.77228093147278
325 9.82059788703918
326 9.85899925231934
327 10.4158360958099
328 9.88954043388367
329 9.92965793609619
330 9.9616117477417
331 9.97774910926819
332 10.0237612724304
333 10.0194852352142
334 10.0621500015259
335 10.2412025928497
336 10.1560118198395
337 10.1640408039093
338 10.1999380588531
339 10.2558035850525
340 10.281628370285
341 10.3152749538422
342 10.3466725349426
343 10.383513212204
344 10.4073872566223
345 10.4545590877533
346 10.4124231338501
347 10.4777572154999
348 10.7562777996063
349 10.5400083065033
350 10.5317075252533
351 10.5707414150238
352 10.6077110767365
353 10.6366217136383
354 10.6676936149597
355 10.699667930603
356 10.9315888881683
357 10.7931272983551
358 10.7937693595886
359 10.8143091201782
360 10.8739545345306
361 10.8814167976379
362 10.9411823749542
363 10.9783697128296
364 11.0040781497955
365 11.4837138652802
366 11.0457100868225
367 11.6758885383606
368 11.1275942325592
369 11.7694327831268
370 11.1770267486572
371 11.1681950092316
372 11.1835839748383
373 11.220207452774
374 11.2584030628204
375 11.2764465808868
376 12.005936384201
377 11.3513443470001
378 11.3569943904877
379 11.410008430481
380 11.4479897022247
381 11.4363222122192
382 11.5700626373291
383 11.5542783737183
384 12.1829500198364
385 11.5869688987732
386 11.6453790664673
387 11.6655030250549
388 11.6820542812347
389 11.742981672287
390 11.7251543998718
391 12.4339218139648
392 11.8400042057037
393 11.886421918869
394 11.9042723178864
395 12.1343297958374
396 11.9556567668915
397 11.9780485630035
398 11.9921982288361
399 12.8514988422394
400 11.9350571632385
401 12.0030777454376
402 12.0361745357513
403 12.7207279205322
404 12.0871515274048
405 12.0703151226044
406 17.3858790397644
407 13.7675397396088
408 12.8468995094299
409 13.2891640663147
410 12.8136835098267
};
\addlegendentry{\colsafe};
\end{axis}

\end{tikzpicture}

%% file: content/conclusions.tex
\section{Conclusions}
\label{sec:conclusions}

We propose a novel algorithm for safe exploration in reinforcement learning.
Under assumptions that are comparable to other safe learning algorithms, such as the popular \safeopt algorithm, we can provide similar high-probability safety guarantees.
However, our algorithm has significantly lower computational complexity, simplifying its use on embedded devices and for high-dimensional systems.
We also see that we trade this lower computational complexity for a more conservative exploration strategy.
Thus, in practice, it depends on the application and the available computational resources which algorithm is preferable.

Herein, we basically adopted the \safeopt algorithm and replaced the Gaussian process estimates with the Nadaraya-Watson estimator. 
For future work, it might be interesting to combine those estimates also with more recent advances of \safeopt such as~\cite{konig2021applied,duivenvoorden2017constrained,sui2018stagewise,sukhija2023gosafeopt}.

Further, we here considered a constant bandwidth parameter $\lambda$.
This parameter could also be tuned to make the algorithm less conservative, and that way, potentially come closer to the performance of \safeopt under the same amount of samples.